\documentclass[11pt]{article}
\usepackage{fullpage}
\usepackage{times}
\usepackage{epsfig}
\usepackage{amsmath,amsthm,amsfonts,amssymb,amscd}

\newtheorem{lemma}{Lemma}
\newtheorem{theorem}[lemma]{Theorem}

\newtheorem{definition}[lemma]{Definition}
\newtheorem{proposition}[lemma]{Proposition}

\newtheorem{conjecture}{Conjecture}

\numberwithin{lemma}{section}

\theoremstyle{definition}

\newcommand\F{{\mathbb{F}}}

\newcommand\Tr{{\mathop\textup{Tr }}}

\newcommand{\set}[1]{\left\{{#1}\right\}}
\newcommand{\poly}{\mathrm{poly}}
\newcommand{\remove}[1]{}

\renewcommand{\epsilon}{\varepsilon}
\newcommand{\eps}{\epsilon}

\renewcommand{\le}{\leqslant}
\renewcommand{\ge}{\geqslant}

\title{Pseudorandom generators and the BQP vs. PH problem}
\author{Bill Fefferman\thanks{Department of Computing and Mathematical Sciences,
Caltech, Pasadena, CA 91125. Supported by IQI.}
\and Chris Umans\thanks{Department of Computing and Mathematical Sciences,
Caltech, Pasadena, CA 91125. Supported by NSF CCF-0846991.}}
\begin{document}

\maketitle
\setcounter{page}{0}
\thispagestyle{empty}

\begin{abstract}
It is a longstanding open problem to devise an oracle relative to which BQP does not lie in the Polynomial-Time Hierarchy (PH).
We advance a natural conjecture about the capacity of the Nisan-Wigderson pseudorandom generator \cite{NW94} to fool $AC_0$, with {\sc majority} as its hard function. Our conjecture is essentially that the loss due to the hybrid argument (which is a component of the standard proof from \cite{NW94}) can be avoided in this setting. This is a question that has been asked previously in the pseudorandomness literature \cite{BSW03}. We then make three main contributions:
\begin{enumerate}
\item We show that our conjecture implies the existence of an oracle relative to which BQP is not in the PH. This entails giving an explicit construction of unitary matrices, realizable by small quantum circuits, whose row-supports are ``nearly-disjoint.''

\item We give a simple framework (generalizing the setting of Aaronson \cite{A}) in which any efficiently quantumly computable unitary gives rise to a distribution that can be distinguished from the uniform distribution by an efficient quantum algorithm. When applied to the unitaries we construct, this framework yields a problem that can be solved quantumly, and which forms the basis for the desired oracle.

\item We prove that Aaronson's ``GLN conjecture'' \cite{A} implies our conjecture; our conjecture is thus formally easier to prove. The GLN conjecture was recently proved false for depth greater than $2$ \cite{A2010}, but it remains open for depth $2$. If true, the depth-2 version of either conjecture would imply an oracle relative to which BQP is not in AM, which is itself an outstanding open problem.

\end{enumerate}
Taken together, our results have the following interesting interpretation: they give an instantiation of the Nisan-Wigderson generator that can be broken by quantum computers, but not by the relevant modes of classical computation, if our conjecture is true.
\end{abstract}

\newpage

\section{Introduction}

Let $U_t$ denote a random variable uniformly distributed on $t$-bit strings. A {\em pseudorandom generator} (PRG) is a function \[f:\{0,1\}^t \rightarrow \{0,1\}^m\] that stretches a short ``seed'' into a longer output string, with the property that $f(U_t)$ is {\em computationally indistinguishable} from the uniform distribution.

There is a vast literature constructing PRGs that achieve computational indistinguishability against a wide variety of computational models (e.g. small circuits, small nondeterministic circuits, small branching programs, small constant-depth circuits). These constructions are typically ``hardness vs. randomness'' tradeoffs in the sense that they make use of a hard function (either unconditionally hard, or hard conditioned on a complexity assumption), and their proof of correctness takes the form of a reduction that transforms a computationally efficient {\em distinguisher} into an efficient algorithm for the hard function (thereby deriving a contradiction). This transformation entails the use of the {\em hybrid argument} \cite{GM84, Y82} which incurs a loss of a factor $1/m$ in going from a distinguisher (with gap $\epsilon$) to a {\em predictor} (with advantage $\epsilon/m$) and from there to an efficient algorithm (with advantage $\epsilon/m)$.

A question that has been raised in the pseudorandomness literature is whether this loss of a factor of $1/m$ can be avoided (for an explicit framing of this question, and a discussion of its motivation, see \cite{BSW03}). In certain settings, the answer is known to be ``yes'' -- when the notion of ``efficient'' is small PH circuits, or bounded-width branching programs \cite{BSW03}. In the present paper, we identify a setting in which this question has surprising connections to a central unresolved question in quantum complexity: whether there exists an oracle relative to which BQP is not in the PH.

Our setting is a familiar one: we will work with the ubiquitous Nisan-Wigderson PRG \cite{NW94}, against $AC_0$ circuits, with {\sc majority} as its hard function. We need a precise statement for the discussion below, which can be given via two standard definitions:

\begin{definition}[\cite{NW94}]
A set family ${\cal D} = \{S_1, S_2, \ldots, S_m\}$ is an $(\ell, p)$ design if every set in the family has cardinality $\ell$, and for all $i \ne j$, $|S_i \cap S_j| \le p$.
\end{definition}

\begin{definition}[\cite{NW94}]
Given a function $f:\{0,1\}^{\ell} \rightarrow \{0,1\}$ and an $(\ell, p)$ design ${\cal D} = \{S_1, S_2, \ldots, S_m\}$ in a universe of size $t$, the function $NW^{f}_{\cal D}:\set{0,1}^t \rightarrow \{0,1\}^m$ is given by
\[NW^{f}_{\cal D}(x) = \left (f_1(x_{|S_1}), f_2(x_{|S_2}), f_3(x_{|S_3}), \ldots, f_m(x_{|S_m})\right ),\]
where each $f_i$ is the function $f$ with a fixed set of its inputs negated\footnote{The standard setup has each $f_i = f$; we need the additional freedom in this paper for technical reasons. We know of no settings in which this alteration affects the analysis of the NW generator.}, and $x_{|S}$ denotes the projection of $x$ to the coordinates in the set $S$.
\end{definition}

Generally speaking, the function $NW_{\cal D}^{f}$ is a PRG against a class of distinguishers as long as $f$ is hard on average for that class of distinguishers. Recall that the majority function on $\ell$ bits is known to be hard for
$AC_0$: no polynomial-size (or even quasi-polynomial-size), constant-depth circuit can compute majority correctly on more than a
$1/2 + \widetilde{O}(1/\sqrt{\ell})$ fraction of the inputs \cite{S93, Has87}, and this is essentially tight, since the function that simply outputs the first bit of the input is correct on a random input with probability $1/2 + \Theta(1/\sqrt{\ell})$. We make the following quantitative conjecture:

\begin{conjecture}
\label{conj:ours} Let ${\mathcal D} = \{S_1, S_2, \ldots, S_m\}$ be an $(\ell, O(1))$-design in a universe of size $t \le
\poly(\ell)$, with $m \le \poly(\ell)$. Then for every constant-depth circuit of size at most $\exp(\poly \log m)$,
\[|\Pr[C(U_{t+m})=1] - \Pr[C(U_t, NW_{\mathcal D}^{\mbox{\small \sc majority}}(U_t))=1]| \le o(1).\]
\end{conjecture}
In this work we abuse notation and refer to constant depth circuits of size at most $\exp(\poly \log m)$ as ``$AC_0$.''

By the standard argument from \cite{NW94, N92}, a distinguishing circuit $C$ with gap $\epsilon$ can be converted to a {\em predictor} with advantage $\epsilon/m$ and then a slightly larger circuit that computes {\sc majority} with success
rate $1/2 + \epsilon/m$. Thus the above statement is true for $m \ll \sqrt{\ell}$; if the $1/m$ loss from the hybrid argument can be avoided (or reduced), it would be true for $m$ as large as $\poly(\ell)$ (and even larger) as we
conjecture is true. In Section \ref{sec:discussion} we discuss intuition supporting this conjecture that relates specifically to the hardness of {\sc majority} for $AC_0$.

This paper contains three main results, which together make Conjecture \ref{conj:ours} interesting and worthy of further study:
\begin{itemize}
\item We show that our conjecture implies the existence of an oracle relative to which BQP is not in the PH, and would thus resolve a major question in quantum complexity. We are encouraged by the fact that our conjecture is recognizable as a natural question in pseudorandomness that has been previously and independently studied (e.g., in \cite{BSW03}).

    The crucial component in showing that our conjecture is sufficient for the existence of an oracle relative to which BQP is not in the PH, is an explicit construction of unitary matrices whose row-supports form an $(\ell, p)$-design. We give such a construction and show how to realize these matrices with small quantum circuits in Section \ref{sec:constructions}. This is the technical core of the paper.

\item We generalize the setting of \cite{A} (which proposed a so-called {\em forrelated} distribution as one that is easy to distinguish from uniform by a quantum computer, but possibly hard for $AC_0$) to a simple framework in which any efficiently quantumly computable unitary $U$ gives rise to a distribution that can be distinguished from uniform by a quantum computer (and Aaronson's setup is recovered by choosing $U$ to be a DFT matrix).

    Together with our construction of explicit unitaries whose row-supports form an $(\ell, p)$-design, this framework has the following interesting interpretation: it gives an instantiation of the Nisan-Wigderson generator that can be broken by quantum computers, but not by the relevant modes of classical computation, if Conjecture \ref{conj:ours} is true.

    Also of independent interest is the fact the unitaries that form the basis of our quantum algorithms don't seem to resemble the DFT matrices for problems in the Hidden Subgroup framework, or even the few other unitaries used in known quantum algorithms. But they possess natural extremal combinatorial (as opposed to algebraic) properties, and we wonder if they can be useful elsewhere in the quantum realm.

\item We show that the ``Nisan-Wigderson'' distribution $(U_t, NW_{\cal D}^{\mbox{\small \sc majority}}(U_t))$ is $\epsilon$-almost $k$-wise independent, in the sense of Aaronson \cite{A}, whose ``GLN conjecture'' asserted that all such distributions fool $AC_0$; a depth-3 counterexample was later found \cite{A2010}. Whether all such distributions fool depth-2 $AC_0$ remains open. A distribution in our general framework (thus efficiently quantumly distinguishable from uniform) that fools depth-2 $AC_0$ would imply an oracle relative to which BQP is not in AM, a weaker (and still unresolved) version of the BQP vs. PH problem. Thus there are two potential routes to resolving this weaker version of the main problem (the depth-2 version of our conjecture, and the depth-2 version of the GLN conjecture); ours is formally easier, and arguably conceptually easier because its connection to the pseudorandomness literature suggests initial lines of attack.
\end{itemize}

Finally, since \cite{A} has shown that the classes $SZK$ and $BPP_{\mbox{path}}$ require exponentially many queries to distinguish $\epsilon$-almost $k$-wise independent distributions from uniform, our constructions {\em unconditionally} yield oracles relative to which BQP does not lie in either of these classes (and $MA$ as well, since $MA \subseteq BPP_{\mbox{path}}$), just as Aaronson's construction does.

\subsection{The BQP vs. PH question}

The quest for an oracle relative to which BQP is not in the PH dates to the foundational papers of the field; the question was first asked by Bernstein and Vazirani \cite{BV93} in the early 1990's. They also gave an oracle problem, {\sc recursive fourier sampling}, that is regarded as a promising candidate (but there have been as yet no real inroads on a potential proof). Currently, oracles are known relative to which BQP is not in MA \cite{W00}, but no relativized worlds are known in which BQP is not in AM. Obtaining an oracle relative to which BQP is not in the PH thus represents a stubborn, longstanding and fundamental problem whose resolution would help clarify the relationship between BQP and classical complexity classes. In recent progress, Aaronson \cite{A} devised a {\em relation} oracle problem that lies in the function version of BQP but not in the function version of the PH, but this still leaves the original problem open. Aaronson's work \cite{A} also has a detailed account of the many motivations for revisiting (and hopefully resolving!) this  problem, and we refer the interested reader to the introduction of \cite{A} for many more details.

In this paper we will find it convenient to speak almost exclusively about the ``scaled down'' version of the problem, which is equivalent via the well-known connection between PH and $AC_0$.  In it, the goal is to design a promise problem (rather than an oracle) that lies in (promise)-BQLOGTIME but not (promise)-$AC_0$
We will drop the cumbersome ``promise'' modifiers when they are clear from context. The class BQLOGTIME is the class of languages decidable by quantum computers that have random access to an $N$-bit input, and use only $O(\log N)$ steps.

\begin{definition}[BQLOGTIME]
A language $L$ is in {\em BQLOGTIME} if it can be decided by a LOGTIME-uniform family of circuits $\{C_n\}$, where each $C_n$ is a quantum circuit on $n$ qubits. On an $(N=2^n)$-bit input $x$, circuit $C_n$ applies $O(\log N)$ gates, with each gate being either a {\em query} gate which applies the map $|i\rangle|z\rangle \mapsto |i\rangle|z \oplus x_i\rangle$, or a standard quantum gate (from a fixed, finite basis).
It is equivalent (by polynomially padding the number of qubits) to allow $\poly\log(N)$ gates.
\end{definition}

Following Aaronson, our goal will be to design, for each input length $N$, a {\em distribution} on $N$-bit strings that can be distinguished from the uniform distribution by a BQLOGTIME predicate, but not by an $AC_0$ circuit. As described in Appendix \ref{app:converting}, such a distribution can be easily converted to a proper oracle $O$ for which $BQP^O \not\subset PH^O$.

\subsection{Techniques}

In this section we briefly discuss the techniques we use for each of the main results listed above.

\paragraph{Showing that our NW distribution is $\epsilon$-almost $k$-wise independent.}

We prove that whenever ${\cal D}$ is an $(\ell, p)$ design in a universe of size $t$, the random variable $(U_t, NW_{\cal D}^{\mbox{\small \sc majority}}(U_t))$ is $O(pk^2/\sqrt{\ell})$-almost $k$-wise independent, for $k < o(\ell^{1/4}p^{-1/2})$. The relevant definition of almost-$k$-wise independence (which we inherit from \cite{A}) appears in Definition \ref{def:almost-k-wise}. Recall that this property of our distribution is the technical basis of the $SZK$ and $BPP_{path}$ results, as well as the connections to the depth-2 GLN conjecture.

This statement amounts to the assertion that after conditioning on the value of up to $k-1$ coordinates, the bias (away from $1/2$) of any specified $k$-th coordinate is at most $O(pk/\sqrt{\ell})$. This is an easy calculation when the conditioned coordinates all lie among the first $t$ coordinates (since the $k$-th coordinate is either completely independent, if it lies among the first $t$ coordinates, or else it is {\sc majority} applied to a subset of $\ell$ of the first $t$ coordinates, of which up to $k-1$ may be fixed). In the actual proof, when some conditioned coordinates lie {\em outside} the first $t$ coordinates (which would otherwise be difficult to analyze), we use the following simple trick to reduce to the easy case: we replace conditioning on coordinate $t+i$ with conditioning on {\em all} of the coordinates in set $S_i$ of the $(\ell, p)$-design (which determine it). Since at most $p$ of these can affect the bias of the $k$-th coordinate, we are back in the easy case with up to $p(k-1)$ bits fixed instead of $(k-1)$.

\paragraph{Showing that our conjecture is sufficient to resolve the BQP vs. PH question.}
In order to show that our conjecture is sufficient to imply an oracle relative to which BQP is not in the PH, we need to discuss the quantum part of the argument. Conjecture \ref{conj:ours} implies that the NW generator with certain parameters fools $AC_0$, which is one part of the overall argument. The other part is to exhibit a BQLOGTIME algorithm that ``breaks'' this instantiation of the NW generator. Generalizing \cite{A}, our quantum algorithm\footnote{We ignore normalization factors in this discussion.} will receive a random string $x \in \{+1,-1\}^t$ (which should be thought of as the input to the NW generator) as the first half of its input, and as the second half of its input, {\em either}
\begin{enumerate}
\item  a second random string in $\{+1, -1\}^t$, {\em or}

\item a string containing the {\em signs} of a unitary $U$ (with entries in $\{0,1,-1\}$) applied to $x$.
\end{enumerate}
The algorithm distinguishes the two cases (roughly) by querying $x$ into the phases, applying $U$, multiplying the second string into the phases, and measuring in the Hadamard basis.

Note that in case (2), each coordinate of the second string is the sign of a $+1/-1$ weighted sum of certain coordinates of $x$; i.e., it computes {\sc majority} (with a fixed pattern of inputs negated) on this subset of the coordinate of $x$. Thus, if we can construct a unitary $U$ whose row-supports form an $(\ell, p)$ design ${\cal D}$ in a universe of size $t$, then case (2) will be the distribution $(U_t, NW_{\cal D}^{\mbox{\small \sc majority}}(U_t))$, and case (1) will be the uniform distribution. The parameters of this instantiation of the NW generator will be such that Conjecture \ref{conj:ours} implies that it fools $AC_0$. Our task becomes to construct such a unitary $U$.

Note that it is {\em not} possible to simply take an existing $(\ell, p)$ design (random, or other explicit constructions that appear in the literature \cite{NW94, HR03}) and attach $+/-$ signs to the elements of the sets so as to make their characteristic vectors pairwise orthogonal, which is what is needed for them to come from the rows of a unitary $U$. On the other hand we have a different setting of the parameters in mind than usual: we want $p$ to be unusually small (a constant), but the number of sets in the design is also unusually small (only $\poly(\ell)$ instead of $\exp(\ell)$). For these parameters we manage to obtain the required $(\ell, p)$ design using a geometric construction, in which the sets are the characteristic vectors of pairs of lines in an affine plane. The strong symmetries in this construction allow us to assign $+/-$ signs to the elements of each set to achieve pairwise orthogonality of their characteristic vectors. In fact these set systems have only $t/2$ (rather than $t$) sets in them, so the resulting unitaries will have the required properties only among half of their rows, but a small modification of the distribution given to the quantum algorithm in case (2) above can handle this without difficulty.

In Section \ref{sec:local-decomposition} we give a {\em local decomposition} (see Section \ref{sec:preliminaries} for the formal definition) of these unitaries, which is necessary to have an {\em efficient} quantum algorithm. This is the most technically involved part of the paper. We also describe a modification of our construction that is {\em extremal} in the sense that it optimizes all relevant parameters simultaneously: {\em all} rows of the unitary participate, we have $p \le 2$, and $t \le \ell^2$. This is not required for our results, but it is aesthetically pleasing. We have been unable to find a local decomposition that would enable us to actually use this construction as the basis of an efficient quantum algorithm, and we leave finding such a decomposition as an intriguing open problem.

\section{NW distributions are $\epsilon$-almost $k$-wise independent}

Aaronson \cite{A} used the following definition of $\epsilon$-almost $k$-wise independence in order to formulate his ``Generalized Linial-Nisan'' (GLN) conjecture.

\begin{definition}
\label{def:almost-k-wise}
A random variable $D$ distributed on $\{0,1\}^r$ is {\em $\epsilon$-almost $k$-wise independent} if for every $k$ distinct indices $i_1, i_2, \ldots, i_k \in [r]$, and every $\alpha_1, \alpha_2, \ldots, \alpha_k \in \{0,1\}$ we have:
\[1 - \epsilon \le \frac{\Pr[D_{i_1} = \alpha_1 \land D_{i_2} = \alpha_2 \land \cdots \land D_{i_k} = \alpha_k]}{2^{-k}} \le 1 + \epsilon.\]
\end{definition}

The following is the GLN conjecture, which is now known to be false for depth 3 and higher \cite{A2010}, but remains open for depth 2:
\begin{conjecture}[\cite{A}]
Let $D$ be any random variable distributed on $\{0,1\}^r$ that is $1/r^{\Omega(1)}$-almost $r^{\Omega(1)}$-wise independent\footnote{One might expect to see $k = \poly\log(r)$ independence rather than $k = r^{\Omega(1)}$, in analogy with the Linial-Nisan conjecture. Aaronson uses the stronger parameter setting (making the GLN conjecture easier) because it is sufficient for his construction; it is for ours too.}. Then for every constant-depth circuit $C$ of size at most $m = 2^{r^{o(1)}}$,
\[|\Pr[C(D)=1] - \Pr[C(U_r)=1]| \le o(1).\]
\label{conj:Scott}
\end{conjecture}

We now show that certain instantiations of the NW generator, including the ones in our Conjecture \ref{conj:ours}, are $\epsilon$-almost $k$-wise independent, with parameters such that the GLN conjecture implies ours.

\begin{theorem}
Let ${\cal D} = \{S_1, S_2, \ldots, S_m\}$ be an $(\ell, p)$ design in a universe of size $t$. Then for every $k < o(\ell^{1/4}p^{-1/2})$, the jointly distributed random variable \[(C,D) = (U_t, NW_{\cal D}^{\mbox{\small \sc majority}}(U_t))\] is $O(pk^2/\sqrt{\ell})$-almost $k$-wise independent.
\label{thm:conj->conj}
\end{theorem}

\begin{proof}
Fix $k_1$ distinct indices $i_1, i_2, \ldots, i_{k_1} \in [t]$ and $k_2$ distinct indices $j_1, j_2, \ldots, j_{k_2} \in [m]$ with $k_1 + k_2 \le k$, and fix $\alpha_1, \alpha_2, \ldots, \alpha_{k_1}, \beta_1, \beta_2, \ldots, \beta_{k_2} \in \{0,1\}$.

We compute the probability
\[\rho = \Pr[{C}_{i_1} = \alpha_1 \land {C}_{i_2} = \alpha_2 \land \cdots \land {C}_{i_{k_1}} = \alpha_{k_1} \land {D}_{j_1} = \beta_1 \land {D}_{j_2} = \beta_2 \land \cdots \land {D}_{j_{k_2}} = \beta_{k_2}],\]
which we write as
\begin{eqnarray*}
\rho & = & \left (\prod_{w=1}^{k_1}\Pr[{C}_{i_w} = \alpha_w | {C}_{i_{1}} = \alpha_{1} \land C_{2} = \alpha_{2} \land \cdots \land C_{i_{w-1}} = \alpha_{i_{w-1}}]\right )\\
& \times & \left (\prod_{w=1}^{k_2}\Pr[D_{j_w} = \beta_j | {C}_{i_{1}} = \alpha_{1} \land C_{2} = \alpha_{2} \land \cdots \land C_{i_{k_1}} = \alpha_{i_{k_1}} \right . \\
& & \hspace{1.5in} \left . \land D_{j_1} = \beta_{j_1} \land D_{j_2} = \beta_{j_2} \land \cdots \land D_{j_{w-1}} = \beta_{w-1}]\right).
\end{eqnarray*}
Clearly the first $k_1$ terms of the product are exactly $1/2$, since $C$ is uniform on $t$-bit strings. Now, consider the $w$-th factor, denoted $\rho_w$, in the second part of the product. The key maneuver is to replace the conditioning on $D_{j_v}$ (for $v < w$)  with conditioning on $D_s$ for $s \in S_w \cap S_v$. This is permissible because $D_{j_v}$ can affect $D_{j_w}$ only through the common elements of their associated sets $S_v$ and $S_w$. Note that because $|S_w \cap S_v| \le p$, the total number of coordinates that are being conditioned upon is $\le pk$.

Recall that $|S_w| = \ell$, and that the bit $D_w$ is the majority (with certain inputs negated) of the specified $\ell$ coordinates of $C$.  Without conditioning, we could compute $\Pr[D_w = 1]$ by
\[\frac{1}{2^\ell} \cdot \sum_{r = \lceil \ell/2\rceil}^{\ell} {\ell \choose r}.\]
We want to compute instead $\rho_w$, which is the same probability conditioned on up to $pk$ of the coordinates of $C$. The maximum value of $\rho_w$ is thus
\[\rho_w \le \frac{1}{2^\ell} \cdot \sum_{r = \lceil \ell/2\rceil - pk}^{\ell} {\ell \choose r}.\]
A simple calculation using Stirling's Approximation shows that ${\ell \choose r} \le O(\frac{2^\ell}{\sqrt{\ell}})$ for all $r$, so we obtain the upper bound of
\[\rho_w \le \frac{1}{2} + O(pk/\sqrt{\ell}).\]
A symmetric argument shows that
\[\rho_w \ge \frac{1}{2} - O(pk/\sqrt{\ell}).\]
Thus we conclude (using that $k < o(\sqrt{\ell}/(pk))$):
\[\rho \le \left (1/2 + O(pk/\sqrt{\ell})\right )^k \le \left [(1/2)\left (1 + O(pk/\sqrt{\ell})\right )\right ]^k\le 2^{-k}\left (1 + O(pk^2/\sqrt{\ell})\right ),\]
and
\[\rho \ge \left (1/2 - O(pk/\sqrt{\ell})\right )^k \ge \left [(1/2)\left (1 - O(pk/\sqrt{\ell})\right )\right ]^k \ge  2^{-k}\left (1 - O(pk^2/\sqrt{\ell})\right ),\]
as required.
\end{proof}

\section{A general framework}

In this section we describe how to turn any efficiently quantumly computable unitary into a distribution that can be distinguished from uniform by a BQLOGTIME machine. Our framework generalizes the setup in \cite{A}. The ``quantum part'' of the paper is almost entirely contained within this section, so we review some relevant preliminaries below before describing the main result.

\subsection{Quantum preliminaries}
\label{sec:preliminaries}

A {\em unitary} matrix is a square matrix $U$ with complex entries such that $UU^{*} = I$, where $U^{*}$ is the conjugate transpose. Equivalently, its rows (and columns) form an orthonormal basis. We name the standard basis vectors of the $N = 2^n$-dimensional vectorspace underlying an $n$-qubit system by $|v\rangle$ for $v \in \{0,1\}^n$. A {\em local} unitary is a unitary that operates only on $b = O(1)$ qubits; i.e. after a suitable renaming of the standard basis by reordering qubits, it is the matrix $U \otimes I_{2^{n-b}}$, where $U$ is a $2^b \times 2^b$ unitary $U$. A local unitary can be applied in a single step of a quantum computer. A {\em local decomposition} of a unitary is a factorization into local unitaries. We say an $N \times N$ unitary is {\em efficiently quantumly computable} if this factorization has at most $\poly(n)$ factors.

A {\em quantum circuit} applies a sequence of local unitaries (``gates'') where each gate is drawn from a fixed, finite set of gates. There are universal finite gate sets for which any efficiently quantumly computable unitary can be realized (up to exponentially small error) by a $\poly(n)$-size quantum circuit \cite{ASV}.

In this paper, the only manner in which our BQLOGTIME algorithm will access the input string $x$ is the following operation, which ``multiplies $x$ into the phases''. There are three steps: (1) query with the query register clean, which applies the map $|i \rangle|0 \rangle \mapsto |i\rangle|0 \oplus x_i \rangle$ (note each $x_i$ is in $\{0,1\}$); (2) apply to the last qubit the map $|0\rangle \mapsto -|0\rangle, |1\rangle \mapsto |1\rangle$; (3) query again to uncompute the last qubit. When we speak of ``multiplying $x$ into the phase'' it will be linguistically convenient to speak about $x$ as a vector with entries from $\{+1, -1\}$, even though one can see from this procedure that the actual input is a $0/1$ vector.

The following lemma will be useful repeatedly. It states (essentially) that a block diagonal matrix, all of whose blocks are efficiently quantumly computable, is itself efficiently quantumly computable. This is trivial when all of the blocks are identical, but not entirely obvious in general. The proof is in  Appendix \ref{app:omitted-proofs}

\begin{lemma}
Fix $N= 2^n$ and $M = 2^m$. Let $U$ be an $N \times N$ block diagonal matrix composed of the blocks $U_1, U_2, \ldots, U_M$, where each $U_i$ is a $N/M \times N/M$ matrix that has a $\poly(n)$-size quantum circuit, a description of which is generated by a uniform $\poly(n)$ time procedure, given input $i$.  Then given three registers of $m$ qubits, $n-m$ qubits, and $\poly(n)$ qubits, respectively, with the third register initialized to $|000\cdots0\rangle$,  there is a $\poly(n)$ size uniform quantum circuit that applies $U$ to the first two registers and leaves the third unchanged.
\label{lem:block-diagonal}
\end{lemma}

\subsection{The quantum algorithm}

Let $A$ be any $N \times N$ matrix with entries\footnote{We could extend this framework to matrices with general entries, but we choose to present this restriction since it is all we need.} in $\{0,1,-1\}$ and pairwise orthogonal rows, and define $S(A, i)$ to be the support of the $i$-th row of matrix $A$. Define $\overline{A}$ to be the matrix $A$ with entries in row $i$ scaled by $1/\sqrt{|S(A, i)|}$, and observe that $\overline{A}$ is a unitary matrix.

Define the random variable $D_{A, M} = (x, z)$ distributed on $\{+1,-1\}^{2N}$ by picking $x \in \{+1, -1\}^N$ uniformly, and setting the next $N$ bits to be $z \in \{+1,-1\}^N$ defined by $z_i = \mbox{sgn}((Ax)_i) = \mbox{sgn}((\overline{A}x)_i)$ for $i \le M$ and $z_i$ independently and uniformly random in $\{+1, -1\}$ for $i > M$.

It will be convenient to think of $M = N$ initially; we analyze the general case because we will eventually need to handle $M = N/2$. Below, we use $U_{2N}$ to denote the random variable uniformly distributed on $\{+1, -1\}^{2N}$.

\begin{theorem}
Let $N = 2^n$ for an integer $n > 0$, and let $M = \Omega(N)$. For every matrix $A \in \{0,1,-1\}^{N \times N}$ with pairwise orthogonal rows, there is a BQLOGTIME algorithm $Q_A$ that distinguishes $D_{A, M}$ from $U_{2N}$; i.e., there is some constant $\epsilon > 0$ for which:
\[|\Pr[Q_A(D_{A,M}) = 1] - \Pr[Q_A(U_{2N}) = 1]| > \epsilon.\]
The algorithm is uniform if $A$ comes from a uniform family of matrices.
\label{thm:alg}
\end{theorem}

\begin{proof}
The input to the algorithm is a pair of strings $x, z \in \{+1, -1\}^N$.

The algorithm performs the following steps:
\begin{enumerate}
\item Enter a uniform superposition $\frac{1}{\sqrt{N}}\sum_{i \in \{0,1\}^n} |i\rangle$ and multiply $x$ into the phase to obtain $\frac{1}{\sqrt{N}}\sum_{i \in \{0,1\}^n}x_i|i\rangle$.

\item Apply $\overline{A}$ to obtain $\frac{1}{\sqrt{N}}\sum_{i \in \{0,1\}^n}(\overline{A}x)_i|i \rangle$.

\item Multiply $z$ into the phase to obtain $\frac{1}{\sqrt{N}}\sum_{i \in \{0,1\}^n}z_i(\overline{A}x)_i|i\rangle$.

\item Define vector $w$ by $w_i = \frac{1}{\sqrt{N}}z_i(\overline{A}x)_i$. Apply the $N \times N$ Hadamard\footnote{This is the matrix $H$ whose rows and columns are indexed by $\{0,1\}^n$, with entry $(i,j)$ equal to $-1^{\langle i, j\rangle}/\sqrt{N}$.} $H$ to obtain $\sum_{i \in \{0,1\}^n}(Hw)_i|i\rangle$, and measure in the computational basis. Accept iff the outcome is $0^n$.
\end{enumerate}

We first argue that the acceptance probability is small in case $(x, z)$ is distributed as $U_{2N}$. This follows from a symmetry argument: for fixed $x$, and $w$ as defined in Step 4 above, the vector $Hw$ above has every entry identically distributed, because $z$ is independently chosen uniformly from $\{-1,+1\}^N$ and
every row of $H$ is a vector in $\{-1, +1\}^N$. In particular this implies that the random variable $(Hw)_i^2$ is identically distributed for all $i$. Together with the fact that $\sum_i (Hw)_i^2 = 1$, we conclude that $E[(Hw)_i^2] = 1/N$. Then by Markov, with probability at least $1 - 1/\sqrt{N}$ we accept with probability at most $\sqrt{N}/N$, for an overall acceptance probability of at most $2/\sqrt{N}$.

Next, we argue that the acceptance probability is large in case $(x, z)$ is distributed as $D_{A, M}$. Here we observe that for $i \le M$, $w_i = \frac{1}{\sqrt{N}}|(\overline{A}x)_i|$ and hence $E[w_i] = \frac{1}{\sqrt{N\cdot|S(A, i)|}}\Omega(\sqrt{|S(A, i)|}) = \Omega(1/\sqrt{N})$ (since before scaling, $w_i$ is just the distance from the origin of a random walk on the line, with $|S(A, i)|$ steps). For $i > M$, we simply have $E[w_i] = 0$. Then $E[\sum_i w_i] = M \cdot \Omega(1/\sqrt{N}) = \Omega(\sqrt{N})$, so $E[(Hw)_{0^n}] = \Omega(1)$. Since the random variable $(Hw)_{0^n}$ is always bounded above by $1$, we can apply Markov to its negation to conclude that with constant probability, it is {\em at least} a constant $\epsilon$ (and in such cases the acceptance probability is at least $\epsilon^2$). Overall, the acceptance probability is $\Omega(1)$.
\end{proof}

The BQLOGTIME algorithm for what Aaronson calls {\sc fourier checking} in \cite{A} is recovered from the above framework by taking $A$ to be a DFT matrix (and $M = N$).

\section{Unitary matrices with large, nearly-disjoint row supports}
\label{sec:constructions}

In this section we construct unitary matrices $A$ with the additional property that all or ``almost all'' of the row supports $S(A, i)$ are large and have bounded intersections. We also show that these unitaries are efficiently quantumly computable. This is the final part of our main result: the distribution $D_{A, M}$ (it will turn out that $M$ will be half the underlying dimension) is distinguishable from uniform by a BQLOGTIME algorithm by Theorem \ref{thm:alg}, and at the same time $D_{A,M}$ can be seen as an NW distribution that by Conjecture \ref{conj:ours} fools $AC_0$ (see Section \ref{sec:putting-it} for the precise statement).

\subsection{The paired-lines construction}
\label{sec:paired-lines}

We describe a collection of $q^2/2$ pairwise-orthogonal rows, each of which is a vector in $\{0,+1,-1\}^{q^2}$. We identify $q^2$ with the affine plane $\F_q \times \F_q$, where $q = 2^n$ for an integer $n > 0$. Let $B_1, B_2$ be an equipartition of $\F_q$, and let $\phi:B_1 \rightarrow B_2$ be an arbitrary bijection. Our vectors are indexed by a pair $(a,b) \in \F_q \times B_1$, and their coordinates are naturally identified with $\F_q \times \F_q$:
\begin{equation}
v_{a,b}[x, y] = \left \{\begin{array}{ll}-1 & y = ax + b\\ +1 & y = ax + \phi(b) \end{array}\right .
\label{eq:paired-lines-construction}
\end{equation}
Notice that $v(a,b)$ is $-1$ on exactly the points of $\F_q \times \F_q$ corresponding to the line with slope $a$ and $y$-intercept $b$, and $+1$ on exactly the points of $\F_q \times \F_q$ corresponding to the line with slope $a$ and $y$-intercept $\phi(b)$. So each $v(a, b)$ is supported on exactly a pair of parallel lines. Orthogonality will follow from the fact that every two non-parallel line-pairs intersect in exactly one point, as argued in the proof of the next lemma.

\begin{lemma}
\label{lem:paired-lines-are-design}
The vectors defined in Eq. (\ref{eq:paired-lines-construction}) are pairwise orthogonal, and their supports form a $(2q, 4)$ design.
\end{lemma}

\begin{proof}
Consider $(a, b) \ne (a', b')$. If $a = a'$ then the supports of $v(a, b)$ and $v(a, b')$ are disjoint. Otherwise $a \ne a'$ and there are exactly four intersection points (obtained by solving linear equations over $\F_q$):
\begin{itemize}
\item $(x = (b' - b)/(a - a'), y = ax + b) = (x = (b' - b)/(a - a'), y = a'x + b')$, which contributes $(-1)\cdot(-1) = 1$ to the inner product, and
\item $(x = (b' - \phi(b))/(a - a'), y = ax + \phi(b)) = (x = (b' - \phi(b))/(a - a'), y = a'x + b')$, which contributes $(+1)\cdot(-1) = -1$ to the inner product, and
\item $(x = (\phi(b') - b)/(a - a'), y = ax + b) = (x = (\phi(b') - b)/(a - a'), y = a'x + \phi(b'))$, which contributes $(-1)\cdot(+1) = -1$ to the inner product, and
\item $(x = (\phi(b') - \phi(b))/(a - a'), y = ax + \phi(b)) = (x = (\phi(b') - \phi(b))/(a - a'), y = a'x + \phi(b'))$, which contributes $(+1)\cdot(+1) = 1$ to the inner product.
\end{itemize}
The sum of the contributions to the inner product from these four points is zero. The computation of the support size is straightforward.
\end{proof}

In Appendix \ref{app:all-rows}, we give another construction (which is not needed for our main result) in which the number of vectors is exactly equal to the dimension of the underlying space (giving rise to a unitary in which ``all rows participate'' instead of only half of the rows).

\subsection{A local decomposition}
\label{sec:local-decomposition}

We new describe an $q^2 \times q^2$ unitary matrix that is efficiently quantumly computable and has the (normalized) vectors $v(a, b)$ from Eq. (\ref{eq:paired-lines-construction}) as $q^2/2$ of its $q^2$ rows. We recall that $q = 2^n$ for an integer $n > 0$.

\begin{proposition}
The following $q \times q$ unitary matrices are efficiently quantumly computable:
\begin{itemize}
\item The DFT matrix $F$ with respect to the additive group of $\F_q$.
\item The inverse DFT matrix $F^{-1}$ with respect to the additive group of $\F_q$.
\item The $q \times q$ unitary matrix $B$ with $\frac{1}{\sqrt{2}}(I_{q/2} | -I_{q/2})$ as its first $q/2$ rows, $\frac{1}{\sqrt{4}}(I_{q/4}|-I_{q/4}|I_{q/4}|-I_{q/4})$ as its next $q/4$ rows, $\frac{1}{\sqrt{8}}(I_{q/8}|-I_{q/8}|I_{q/8}|-I_{q/8}|I_{q/8}|-I_{q/8}|I_{q/8}|-I_{q/8})$ as its next $q/8$ rows, etc... and whose last row is $\frac{1}{\sqrt{N}}(1,1, 1, \ldots, 1).$
\end{itemize}
\end{proposition}

\begin{proof}
The first two matrices are well-known to be efficiently quantumly computable. For the last one we make use of the Hadamard matrix
\[H = \frac{1}{\sqrt{2}}\left (\begin{array}{cc} 1 & -1 \\ 1 & 1 \end{array} \right ).\]
Let $B_i$ be the $q \times q$ identity matrix with its lower right $2^i \times 2^i$ submatrix replaced by the matrix $H \otimes I_{2^{i-1}}$. Each $B_i$ is efficiently quantumly computable by Lemma \ref{lem:block-diagonal}. It is then easy to verify that
$B = B_1B_2B_3\cdots B_n.$
\end{proof}

\begin{lemma}
Let $\alpha$ be a generator of the multiplicative group of  $\F_q$. For $c \in \F_q$, let $D_c$ denote the $q \times q$ diagonal matrix
\[\frac{1}{\sqrt{q}} \cdot \mbox{diag}\left (\sqrt{q}, (-1)^{\Tr(\alpha^1\cdot c)}, (-1)^{\Tr(\alpha^2\cdot c)}, (-1)^{\Tr(\alpha^3\cdot c)}, \ldots, (-1)^{\Tr(\alpha^{q-1}\cdot c)}\right ),\]
and let $D_c'$ denote the $q \times q$ diagonal matrix
\[\frac{1}{\sqrt{q}}\cdot \mbox{diag}\left (0, (-1)^{\Tr(\alpha^1\cdot c)}, (-1)^{\Tr(\alpha^2\cdot c)}, (-1)^{\Tr(\alpha^3\cdot c)}, \ldots, (-1)^{\Tr(\alpha^{q-1}\cdot c)}\right ).\]
Then the $q^2 \times q^2$ matrix $D$ whose $(i, j)$ block (with $i, j \in \F_q$) equals $D_{ij}$ if $i = j$ and $D_{ij}'$ otherwise, is efficiently quantumly computable.
\end{lemma}

\begin{proof}
Consider the $q^2 \times q^2$ block-diagonal matrix that has as its $(k, k)$ block the matrix whose $(i, j)$ entry is $(-1)^{\Tr{(ij\alpha^k)}}$ for $k \in \set{1,2,\ldots, q-1}$ and whose $(0,0)$ block is $I_q$. Each such block except the $(0,0)$ block is the DFT matrix $F$ with its rows (or equivalently, columns) renamed according to the map $j \mapsto j\alpha^k$. The $F$ matrix is efficiently quantumly computable and the map $j \mapsto j\alpha^k$ is classically and reversibly (and thus quantumly) efficiently computable. Thus each $q \times q$ block on the diagonal is efficiently quantumly computable.  By Lemma \ref{lem:block-diagonal} the entire matrix is efficiently quantumly computable.

If we index columns by $(i, i') \in (\F_q)^2$ and rows by $(j, j') \in (\F_q)^2$, then the desired matrix $D$ is the above block-diagonal matrix with the order of the two indexing coordinates for the rows transposed, and the order of the two indexing coordinates for the columns transposed.
\end{proof}

\begin{theorem}
\label{thm:local-unitary-construction}
The $q^2 \times q^2$ matrix $(I_q \otimes B)\cdot (I_q \otimes F) \cdot D \cdot (I_q \otimes F^{-1})$, which is efficiently quantumly computable, has the vectors $v(a, b)$ from Eq. (\ref{eq:paired-lines-construction}) as $q^2/2$ of its rows\footnote{To be precise, these are the $v(a, b)$ with respect to {\em some} equipartition $B_1, B_2$ and {\em some} bijection $\phi$.}.
\end{theorem}

\begin{proof}
Let $S_c$ be the $q \times q$ permutation matrix $S_c$ that (when multiplied on the right) shifts columns, identified with $\F_q$, by the map $x \mapsto x+c$. Let $J$ be the all-ones matrix. The main observation is that \[FD_cF^{-1} = \frac{1}{\sqrt{q}}S_c - \frac{\sqrt{q} - 1}{q}J,\] and that \[FD_c'F^{-1} = \frac{1}{\sqrt{q}}S_c - \frac{1}{\sqrt{q}}J.\] Thus the final matrix has in its $(i,j)$ block (with $i,j \in \F_q$) the matrix
\[B\cdot\left (\frac{1}{\sqrt{q}}S_{ij} -\frac{\sqrt{q}-1}{q}J\right )\]
if $i = j$, and
\[B\cdot\left (\frac{1}{\sqrt{q}}S_{ij} - \frac{1}{\sqrt{q}}J\right )\]
otherwise. Observe that $BJ$ has all zero entries except for the last row, so in particular, the first $q/2$ rows of the $(i, j)$ block are $(1/\sqrt{2q})(I_{q/2} | -I_{q/2})S_{ij}$. Therefore the $q/2$ rows of the entire $q^2 \times q^2$ matrix corresponding to the top halves of blocks $(i, j)$ as $j$ varies, give the vectors $v(i, b)$ for $b \in B_1$, if we identify columns with $\F_q \times \F_q$ as follows: columns of the $j$-th block are identified with $\{j\} \times \F_q$, and within the $j$-th block, $B_1$ is the first $q/2$ columns and $B_2$ is the next $q/2$ columns (and the bijection $\phi$ maps the element associated with the $b$-th column to the element associated with the $(b+q/2)$-th column).

Then, as $i$ varies over $\F_q$, we find all of the vectors from Eq. (\ref{eq:paired-lines-construction}) as the ``top-halves'' of each successive set of $q$ rows of the large matrix.
\end{proof}

\section{Putting it all together}
\label{sec:putting-it}

Let $A$ be the matrix of Theorem \ref{thm:local-unitary-construction}, and set $N = q^2$ and $M = N/2$. By Theorem \ref{thm:alg}, there is a BQLOGTIME algorithm that distinguishes $D_{A, M}$ from the the uniform distribution $U_{2N}$.

By Lemma \ref{lem:paired-lines-are-design}, the first $M$ rows of $A$ have supports forming a $(2\sqrt{N}, 4)$-design $\cal D$. It is also clear that for $i \le M$, the $(N + i)$-th bit of $D_{A, M}$ computes {\sc majority} (with a fixed pattern of inputs negated) on those among the first $N$ bits that lie in $S(A, i)$. Thus $D_{A, M}$ is exactly the distribution $(U_N, NW_{\cal D}^{\mbox{\small \sc majority}}(U_N))$ followed by $N/2$ additional independent random bits (which can have no impact on the distinguishability of the distribution from uniform). Thus by Conjecture \ref{conj:ours}, no constant-depth, polynomial-size circuit can distinguish $D_{A, M}$ from $U_{2N}$, which completes the argument.

We briefly describe why the standard NW argument fails (and why we must rely on Conjecture \ref{conj:ours}). The standard argument proceeds as follows: define $2N+1$ hybrid distributions $D_{A, M} = H_0, H_1, \ldots, H_{2N} = U_{2N}$, that interpolate between $D_{A, M}$ and $U_{2N}$. Given a distinguishing circuit $C:\{0,1\}^{2N} \rightarrow \{0,1\}$ for which
\[|\Pr[C(D_{A, M}) = 1] - \Pr[C(U_{2N})=1]| \ge \epsilon,\]
we argue that for some $i$
\[|\Pr[C(H_i) = 1] - \Pr[C(H_{i+1})=1]| \ge \epsilon/M\]
by the triangle inequality (and here we are making the additional observation that $H_0 = H_1 = \cdots =  H_{N}$ and $H_{N + M+1} = H_{N + M+2} = \cdots = H_{2N}$ so the gap of $\epsilon$ must be spread over only $M$ differences). From here, we obtain a next-bit-predictor with advantage $\epsilon/M$ and hardwire at most $M$ lookup tables of size $2^p$, to obtain a circuit of size $|C| + O(2N) + O(2^pM)$ that computes {\sc majority} (on $2\sqrt{N}$ bits) with success probability $1/2 + \epsilon/M$. The problem is that this advantage over random guessing is not sufficient to obtain a contradiction for the function {\sc majority}, which can be computed easily with success probability $1/2 + \Omega(N^{1/4})$, for the parameters coming from the unitary $A$ from Theorem \ref{thm:local-unitary-construction}.

Even if we had a unitary whose rows formed an $(\ell, p)$-design with better parameters, the standard argument fails. This is because it must be that $\ell \le N$, and yet we must also have $M \gg \sqrt{N}$ for $D_{A, M}$ to be even {\em statistically} noticably different from uniform. But the trivial circuit that outputs an arbitrary bit of the input succeeds with probability $1/2 + \Omega(1/\sqrt{\ell})$ which is larger than the $1/2 + \epsilon/M$ that comes out of the standard NW argument above.

\section{Our conjecture: discussion}
\label{sec:discussion}

We believe that Conjecture \ref{conj:ours} is quite approachable, given the large literature and variety of proof techniques concerning pseudorandom generators and related objects. As examples, we mention two ideas from the literature that seem relevant (although obviously they haven't yet led to a solution).

The first is the analysis by Sudan, Trevisan, and Vadhan \cite{STV01} of the NW PRG when applied to a ``mildly hard'' predicate (i.e., one for which small circuits fail on only a $\delta$ fraction of the inputs). They prove that the output distribution is computationally indistinguishable from a distribution having high entropy by invoking Impagliazzo's hard-core lemma \cite{I95}, and arguing that output bits of the NW PRG ``often'' fall in a hard core that is considerably harder on average than the original mildly hard predicate.

We also have a hard predicate whose average-case hardness falls short of what we would need for Conjecture \ref{conj:ours} to be true via the standard argument; i.e., if {\sc majority} on $\ell$ bits were $1/2 + 1/\poly(\ell)$ hard, we would be done. The high-level message of Sudan, Trevisan and Vadhan is that this hardness can be achieved (essentially) at the price of comparing to a high-entropy distribution rather than the uniform distribution. Our BQP algorithm is fairly robust and would likely still work on a sufficiently high entropy distribution (it is only necessary to ``kill'' correlations with a particular element of the Hadamard basis). However, the central technical component of the proof in \cite{STV01} is the Impagliazzo hard-core lemma \cite{I95}, and a sufficiently strong hardcore lemma is not known for $AC_0$. In fact, the function {\sc majority} has no hard core:
\begin{proposition}
No subset of {\sc majority} is $\epsilon$-hardcore for $AC_0$, for any $\epsilon < 1/n$.
\end{proposition}
\begin{proof}
Given a $x \in \{0,1\}^n$, the randomized procedure that picks a random one of the $n$ input bits and outputs it succeeds in computing $\mbox{\sc majority}(x)$ with probability at least $1/2 + 1/n$. This procedure has the same success probability over any subset $S \subseteq \{0,1\}^n$. For any fixed $S$, there is a fixing of the random bits that preserves this success probability, and which results in a circuit of size $1$ (it just outputs $x_i$ for some fixed $i$).
\end{proof}
\noindent Nevertheless, it may be that replacing the uniform distribution with a high minentropy one can be useful in circumventing the loss from the hybrid argument.

The second approach is to directly circumvent the loss due to the hybrid argument. This is explicitly addressed in \cite{BSW03}, where they show that the loss can indeed be avoided in certain computational models. One of these models is ``PH circuits,'' which sounds superficially like it might be relevant to our setting. What is actually needed to use their ideas is the ability to approximately count an efficiently recognizable set, in the same class that recognizes the set. Such a statement is not known (or expected) for $AC_0$, but it is still possible that other ideas could circumvent the hybrid argument for $AC_0$.

However, any route to proving Conjecture \ref{conj:ours} faces the same challenge discussed in \cite{A}: the proof must be ``non-black-box'' in the sense that it can't apply to arbitrary low-degree polynomial functions in addition to its native Boolean setting. This is because the quantum algorithm of Theorem \ref{thm:alg} implies (via \cite{BBCMW01}) the existence of a constant-degree, multivariate real polynomial computing the acceptance probability (and hence distinguishing the NW distribution from uniform). A black-box reduction would transform a distinguisher of this form to a similarly low-degree polynomial approximating {\sc majority}, but we know that no such polynomial for approximating {\sc majority} can exist \cite{S93}. So any proof of Conjecture \ref{conj:ours} must prove that the distribution in question fools $AC_0$ in some way that does {\em not} replace $AC_0$ circuits by low degree approximating polynomials and then argue about those.

Here are some ideas that could plausibly form the basis of a proof of Conjecture \ref{conj:ours}. We consider the simpler situation in which the distributions being compared are $N^2$ independent copies of the random variable $D$ -- where $D = (U_N, {\mbox{\sc majority}}(U_N))$  -- and $N^2$ independent copies of the random variable $U_{N+1}$ distributed uniformly on $N+1$ bits. This corresponds to the NW construction we have been working with, if the underlying nearly-disjoint sets are taken to be {\em completely disjoint}. $AC_0$ should be incapable of distinguishing these distributions; here is the intuition. First, observe that there are no correlations between blocks, so the hypothetical distinguisher must examine each block separately. Since $AC_0$ cannot approximate majority well, we know that the only ``accessible'' information about each block is a ``noisy bit'' saying whether it is distributed according to $D$ or $U_{N+1}$ -- in the case of uniform, this bit is $1$ with probability $1/2$, and in the case of distribution $D$, this bit is $1$ with probability $1/2 + \Theta(1/\sqrt{N})$. How can a hypothetical distinguisher aggregate these noisy bits across the $N^2$ independent copies? In one case, the expected sum of these noisy bits $(1/2)N^2$ and in the other case it is $(1/2 + \Theta(1/{\sqrt{N}}))N^2$, and by concentration of measure, the sum is highly likely to be close to these expectations. So the hypothetical distinguisher only needs to tell the difference between $N^2$ fair coin flips versus $N^2$ slightly biased coin flips. But exactly this task is hard for $AC_0$ (which can be seen by reduction from {\sc majority}, as written down in Corollary 12 of \cite{A}). So, it seems that either the distinguisher must approximate {\sc majority} better than allowed (to get less noisy bits), or it must be detecting very small bias in a sequence of coin
flips. In upcoming work \cite{FSUV10}, we are able to show that indeed $AC_0$ cannot distinguish these two distributions. This is encouraging because it shows that the aforementioned ``non-black-box'' requirement is not insurmountable. Extending this result to the not-completely-disjoint case still seems challenging, however.

\medskip
\medskip

\paragraph{Acknowledgements.} We thank Scott Aaronson, Yi-Kai Liu, and Emanuele Viola for helpful discussions.

\newcommand{\etalchar}[1]{$^{#1}$}

\appendix
\section{Omitted proofs}
\label{app:omitted-proofs}

\begin{proof}(Of Lemma \ref{lem:block-diagonal})
Fix a finite universal set of quantum gates, of cardinality $d$, each of which operates on at most $b$ qubits. A convenient notion will be that of an {\em oblivious} circuit, in which we fix an ordering (say, lexicographic) on $[n]^b$, and the steps of the circuit are identified with $\poly(n)$ cycles through this list: when we are on step $(a_1, a_2, \ldots, a_b) \in [n]^b$ in one of these cycles, we operate on qubits $a_1, a_2, \ldots, a_b$. Clearly, any (uniform) quantum circuit can be converted to a (uniform) ``oblivious'' circuit with at most an $n^b$ blowup by inserting dummy identity gates.

Let $n^k$ be an upper bound on the size of the oblivious circuits obtained in this way for the various $U_i$. The circuit for each $U_i$ is now a sequence \[j^{(i)} = \left (j^{(i)}_1, j^{(i)}_2, j^{(i)}_3, \ldots, j^{(i)}_{n^k}\right ),\]
with each $j^{(i)}_\ell \in [d]$ specifying which gate to apply at step $\ell$ in the oblivious circuit for $U_i$ (and because the circuit is oblivious, the qubits to which this gate should be applied are easily determined from $\ell$). Let $f:[M] \rightarrow [d]^{n^k}$ be the function that maps $i$ to the vector $j^{(i)}$.

Now we describe the promised efficient quantum procedure:
\begin{enumerate}
\item Apply the map derived from $f$ that takes $|i\rangle |z \rangle$ to $|i\rangle|z \oplus f(i)\rangle$, to the first and third register. We view the contents of the third register as a vector in $[d]^{n^k}$.

\item Repeat for $\ell = 1,2, 3, \ldots, n^k$: apply the ``controlled unitary'' that consults the $\ell$-th component of the third register, and applies the specified gate to qubits $(a_1, a_2, \ldots, a_b)$ of the second register (again, $(a_1, a_2, \ldots, a_b)$ are easily determined from $\ell$ because the circuit is oblivious). The important observation is that this ``controlled unitary'' operates on only constantly many qubits.

\item Repeat step 1 to uncompute the auxiliary information in the third register.
\end{enumerate}
\end{proof}

\section{A unitary in which all rows participate}
\label{app:all-rows}

There is a tension between the triple goals of (1) having many pairwise orthogonal vectors, (2) maintaining bounded pairwise intersections of the supports, and (3) having the supports large. It is natural to wonder whether the above construction (in which we found a number of vectors equal to $1/2$ the dimension of the underlying space) is in some sense optimal. For example, is there some barrier to simultaneously optimizing all three goals?

Here we show that one can indeed optimize all three goals at the same time, by specifying a construction that builds on the ``paired-lines'' construction. Our construction will have as many pairwise orthogonal vectors as the dimension of the underlying space (which is obviously as many as is possible); it will have intersections sizes bounded above by $2$ (the upper bound cannot be 0 without constraining the product of the number of rows and the support sizes to be at most the dimension of the underlying space, and no pairwise intersections can have cardinality one without violating orthogonality); the support sizes will be at least the square root of the dimension of the underlying space (and one can't exceed that without having larger intersection sizes).

This construction is not needed for our main results, but we find it aesthetically pleasing that one can optimize all three parameters in this way. We {\em don't} know of a local decomposition for this matrix, and we leave finding one as an intriguing open problem.

While the construction of Section \ref{sec:paired-lines} needed characteristic two, the present construction needs odd characteristic. We fix $\F_q$ with $q$ an odd prime power, and we choose a subset $Q \subseteq \F_q^*$ of size $(q-1)/2$ for which $Q \cap -Q = \emptyset$, where $-Q = \{-x : x \in Q\}$. Our vectors will have $q^2 - 1$ coordinates, identified with the {\em punctured plane} $P = \F_q \times \F_q \setminus \{(0,0)\}$.

We have three types of vectors in $\{0,-1,+1\}^P$: first, for all $a \in \F_q$ and $b \in Q$
\begin{equation}
v_{a, b}[x,y] = \left \{ \begin{array}{ll}
+1 & x = 0, y = b \\
+1 & x \in Q, y = ax + b \\
-1 & x \in Q, y = ax - b \\
0 & \mbox{otherwise}
\end{array}\right .,
\label{eq:type-1}
\end{equation}
second, for all $a \in \F_q$ and $b \in -Q$
\begin{equation}
v_{a, b}[x,y] = \left \{ \begin{array}{ll}
+1 & x = 0, y = b \\
+1 & x \in -Q, y = ax + b \\
-1 & x \in -Q, y = ax - b \\
0 & \mbox{otherwise}
\end{array}\right .,
\label{eq:type-2}
\end{equation}
and finally, for each $c \in \F_q^*$
\begin{equation}
u_{c}[x,y] = \left \{ \begin{array}{ll}
+1 & x = c, y \in \F_q \\
0 & \mbox{otherwise}
\end{array}\right . .
\label{eq:type-3}
\end{equation}

\begin{lemma}
The vectors defined in Eqs. (\ref{eq:type-1}), (\ref{eq:type-2}) and (\ref{eq:type-3}) are pairwise orthogonal and their supports form a $(q, 2)$-design.
\end{lemma}

\begin{proof}
It is an easy computation to see that the support of each of the vectors has cardinality $q$. We now argue that they are pairwise orthogonal. There are several cases depending on the two rows under consideration:
\begin{enumerate}
\item $v_{a,b}$ and  $v_{a', b'}$: if one comes from Eq. (\ref{eq:type-1}) and the other from Eq. (\ref{eq:type-2}) then the supports are disjoint. So we assume both come from Eq. (\ref{eq:type-1}) or both come from Eq. (\ref{eq:type-2}).
\begin{enumerate}
    \item Both come from Eq. (\ref{eq:type-1}) and $b = b'$: we have one intersection $(0, b)$ (which contributes $+1$ to the inner product) and exactly one of the following two intersection points: $(x = -2b/(a - a'), ax+b = a'x - b)$ or $(x = 2b/(a - a'), ax-b = a'x +b)$, which contributes $-1$ to the inner product. We have exactly one because the two $x$-values are negations of each other, and non-zero, so exactly one is in $Q$. \label{case:1a}
    \item Both come from Eq. (\ref{eq:type-1}) and $b \ne b'$: we have exactly one of the following two intersection points: $(x = (b' - b)/(a - a'), ax+b = a'x + b')$ or $(x = (-b' + b)/(a - a'), ax-b = a'x-b')$, which contributes $+1$ to the inner product, and exactly one of the following two intersection points: $(x = (b' + b)/(a - a'), ax-b = a'x + b')$ or $(x = (-b' - b)/(a - a'), ax+b = a'x-b')$, which contributes $-1$ to the inner product. For each pair, there is exactly one of the pair of possible intersection points because the two $x$-values are negations of each other, and non-zero, so exactly one is in $Q$. \label{case:1b}
    \item Both come from Eq. (\ref{eq:type-2}) and $b = b'$: identical to case (\ref{case:1a}) above, with $-Q$ in place of $Q$.
    \item Both come from Eq. (\ref{eq:type-2}) and $b \ne b'$: identical to case (\ref{case:1b}) above, with $-Q$ in place of $Q$.
\end{enumerate}
\item $u_c$ and $u_c'$: these have disjoint supports for $c \ne c'$.
\item $v_{a,b}$ and $u_c$: if $c \in Q$, then the support of $u_c$ intersects the support of $v_{a,b}$ only if $v_{a,b}$ comes from Eq. (\ref{eq:type-1}), and then we get one intersection at point $(x = c, ax + b)$ which contributes a $+1$ to the inner product, and one intersection at point $(x=c, ax-b)$ which contributes a $-1$ to the inner product. If $c \in Q$, then the support of $u_c$ intersects the support of $v_{a,b}$ only if $v_{a,b}$ comes from Eq. (\ref{eq:type-2}), and we have an identical argument, with $-Q$ in place of $Q$.
\end{enumerate}
This is a complete enumeration of cases, and in no case did we have more than 2 intersection points.
\end{proof}

We conclude this section with a question: are these matrices related in some way to the DFT matrix over some family of non-abelian groups (e.g. the affine group $\F_q^* \ltimes \F_q$), or are they indeed completely different from the unitaries seen before in quantum algorithms?

\section{Converting a distributional oracle problem into a standard oracle}
\label{app:converting}

We include this section for completeness, a similar proof appears in \cite{A2010}.

We have two ensembles of random variables $D_1 = \{D_{1, n}\}, D_2 = \{D_{2, n}\}$ over $(N = 2^n)$-bit strings for which BQLOGTIME can distinguish the two distributions but $AC_{0}$ cannot. Then when $D_1$ and $D_2$ are viewed as distributions on (truth-tables of) {\em oracles}, there is a BQP oracle machine that distinguishes the two distributions, but no $PH$ oracle machine can distinguish them. Specifically, we have that there exists a BQP oracle machine $A$ for which
\[\Pr[A^{D_1}(1^n)=1] - \Pr[A^{D_2}(1^n)=1] \ge \epsilon\]
while for every $PH$ oracle machine $M$,
\[\Pr[M^{D_1}(1^n)=1] - \Pr[M^{D_2}(1^n)=1] \le \delta < \epsilon,\]
(here we use standard techniques -- see, e.g., \cite{Has87} -- which show that on any fixed input, the output of a $PH$ machine as a function of the oracle can be seen as an $AC_{0}$\footnote{Recall that we are using ``$AC_0$'' to refer to size $\exp(\poly \log n)$-size constant depth circuits in this paper.} circuit)
and we have $\epsilon > \delta$ for sufficiently large $n \ge n_0$.

We now convert the distributions on oracles into a single oracle $O$ for which $BQP^O \not\subseteq PH^O$. Let $L$ be a uniformly random unary language in $\{1\}^*$.  For each $n$, if $1^n \in L$, sample a $2^n$-bit string $x$ from $D_1$ and define oracle $O$ restricted to length $n$ so that $x$ is its truth table; otherwise sample a $2^n$-bit string $x$ from $D_2$ and define oracle $O$ restricted to length $n$ so that $x$ is its truth table.

First, note that
\begin{eqnarray*}
\Pr[A^O(1^n) = L(1^n)] & = & (1/2)\cdot \Pr[A^{D_1}(1^n) = 1] + (1/2)\cdot \Pr[A^{D_2}(1^n) = 0] \ge 1/2 + \epsilon/2.
\end{eqnarray*}

Now fix any PH machine $M$, and note that for sufficiently large $n$,
\begin{eqnarray*}
\Pr[M^O(1^n) = L(1^n)] & = & (1/2)\cdot \Pr[M^{D_1}(1^n) = 1] + (1/2)\cdot \Pr[M^{D_2}(1^n) = 0] \le 1/2 + \delta/2.
\end{eqnarray*}

Consequently, since $\eps > \delta$ there is a fixed choice for the oracle at length $n$ such that $L(1^n) = A^O(1^n) \ne M^O(1^n)$, for sufficiently large $n$.

Fix such a choice for the oracle at length $n$, and consider another PH machine $M'$. By the same argument, we can find another sufficiently large input length $n'$ where $L(1^{n'}) = A^O(1^{n'}) \ne M^O(1^{n'})$.\footnote{We have assumed that our machines, on an input of length $n$, only query the oracle at inputs of length $n$. This can be ensured by working with input lengths that are sufficiently spread out (so that the machine cannot afford to formulate queries to the next largest length, and so that the oracle at shorter lengths can be hardcoded.)}

Continuing in this way, we obtain a single oracle such that for any PH machine $M$ there exists some $n$ for which $A^O(1^n) \ne M^O(1^n)$.

\end{document}